\documentclass[letterpaper, 10 pt, conference]{ieeeconf}  
\IEEEoverridecommandlockouts
\usepackage{hyperref}  

\usepackage{amsmath,amssymb,amsfonts}%
\usepackage{theorem}%
\usepackage{graphicx,color}%
\usepackage{subcaption}
\usepackage{epstopdf}
\usepackage{psfrag}

\newtheorem{theorem}{Theorem}
\newtheorem{definition}{Definition}
\newtheorem{corollary}{Corollary}
\newtheorem{lemma}{Lemma}
\newtheorem{assumption}{Assumption}
\newtheorem{proposition}{Proposition}
\newtheorem{remark}{Remark}

\def\field#1{\mathbb #1}%
\def\R{\field{R}}%

\allowdisplaybreaks%

\def\red{\textcolor{red}}

\newcommand \sgn   {\text{sgn}}

\newcommand{\normt}[1]{{\left\vert\kern-0.25ex\left\vert\kern-0.25ex\left\vert #1 
    \right\vert\kern-0.25ex\right\vert\kern-0.25ex\right\vert}}

\title{Higher Order Convergent Control Barrier Functions for Leader-Follower Multi-Agent Systems under STL Tasks}%

\author{Maryam Sharifi and
	Dimos~V.~Dimarogonas \thanks{This work was supported by the ERC CoG LEAFHOUND, the Swedish Research
Council (VR) and the Knut \& Alice Wallenberg Foundation (KAW).
	The authors are with Division of Decision and Control Systems, School of Electrical Engineering and Computer Science, KTH Royal Institute of Technology, Stockholm, Sweden. \texttt{\{msharifi, dimos\}@kth.se.}}
}
\date{}%

\begin{document}

\maketitle%
\begin{abstract}
This paper presents control strategies based on time-varying convergent higher order control barrier functions for a class of leader-follower multi-agent systems under signal temporal logic (STL) tasks. Each agent is assigned a local STL task which may be dependent on the behavior of agents involved in other tasks. 
The leader has knowledge on the associated tasks and controls the performance of the subgroup involved agents. Robust solutions for the task satisfaction, based on the leader's accessibility to the follower agents' states are suggested. 
Our approach finds solutions to guarantee the satisfaction of STL tasks independent of the agents' initial conditions. 
\end{abstract}
\section{Introduction}
Improved capabilities of coordination in a group of systems over single-agent systems to handle task complexity and robustness to agent failures, makes the field of multi-agent systems a popular research topic.
However, many complex tasks may not be defined as stand-alone traditional control objectives and need employing some tools from computer science such as formal verification in order to define general task specifications in temporal logic formulations that induce a sequence of control actions \cite{kloetzer2009automatic}.
Among those formulations, signal temporal logic (STL) is more beneficial as it is interpreted over continuous-time signals \cite{maler2004monitoring}, allows for imposing tasks with strict deadlines and introduces quantitative robust semantics
\cite{fainekos2009robustness}.

Leader-follower approaches, where a subset of agents are responsible for guiding the whole group to satisfy STL tasks, will contribute to important attributes of multi-agent systems such as scalability. Some recent researches in the leader-follower framework have focused on leader selection for optimal performance \cite{fitch2016optimal} or prescribed performance  control strategies \cite{chen2020leader}. However, they don't take into account complex tasks with space and time constraints prescribed by STL. 

We present control strategies for first and second order leader-follower networks under local STL tasks.
For this aim, we present a notion of \emph{time-varying convergent higher order control barrier functions (TCHCBF)} to address the high relative degree constraints in the case of second order agent dynamics.
Control barrier functions \cite{ames2016control} guarantee the existence of a control law that renders a desired set forward invariant. Nonsmooth and higher order control barrier functions are provided in \cite{glotfelter2017nonsmooth} and \cite{xiao2019control}, respectively. 
Nevertheless, appropriate control barrier functions to maintain the desired behavior of leader-follower multi-agent systems under STL tasks haven't been introduced yet, to the best of our knowledge.
We consider connected graph topologies where each local STL task is defined on a subset of connected agents containing one leader. 
The leader agent has the knowledge of the associated local task and is responsible for its satisfaction. 
The followers are not aware of the prescribed tasks and don't have any control authority to meet them. 

We first consider the case of first order dynamics leader-follower networks. Due to the deficiencies in the rank of the input matrix, there exist singularities in the associated constraints. This issue results from the under-actuated property of the system caused by the follower agents which are not influenced by direct actuation.
We tackle the singularities by providing novel barrier function certificates for specific graph topologies to guarantee fixed-time convergence to the specified safe sets and remaining there onwards. We call these sets \emph{fixed-time convergent and forward invariant}.
We then consider second order leader-follower networks, where the relative-degree of each agent is two. Moreover, there exist again singularities which cause infeasibilties in the satisfaction of required constraints due to the existence of follower agents. We provide higher order convergent control barrier functions and singularity avoidance solutions to satisfy specifications.
In this paper, we extend our previous work \cite{sharifi2021} for the framework of leader-follower networks, where control barrier certificates for first and second order dynamics leader-follower networks based on the knowledge of the leader from the followers are provided, which guarantee convergence and forward invariance of the desired sets. 
We provide relaxed barrier certificates for the input signal in the presence of partial knowledge of the leader from the network, while there is no need for the leader to know the upper bound of the norm corresponding to the dynamic terms of non-neighbor agents. This upper-bound determines the ultimate convergent set for the network under the specified tasks.


The rest of the paper is organized as follows. Section \ref{setup} gives some preliminaries on STL, leader-follower multi-agent systems and time-varying barrier functions. First order systems are considered in Sections \ref{solution3}. Higher order leader-follower networks are considered in Section \ref{solution2}. Simulation results and some concluding points are presented in Sections \ref{sim} and \ref{conc}, respectively.
\section{Preliminaries and problem formulation}\label{setup}
\subsection{Signal temporal logic (STL)}
Signal temporal logic (STL) \cite{maler2004monitoring} is based on predicates $\nu$ which are obtained by evaluation of a continuously differentiable predicate function $h: \mathbb{R}^{d}\to\mathbb{R}$ as $\nu:=\top$ (True) if $h(\mathbf{x})\geq 0$ and $\nu:=\bot$ (False) if $h(\mathbf{x})< 0$ for $\mathbf{x}\in\mathbb{R}^{d}$. The STL syntax is then given by
\begin{align*}
\phi ::=\top |\nu|\neg\phi|\phi' \wedge {\phi ''}| \phi' U_{\left[ {a,b} \right]}{\phi ''},
\end{align*}
where $\phi'$ and $\phi ''$ are STL formulas and $U_{\left[ {a,b} \right]}$ is the until operator with $a\leq b<\infty$. In addition, define $F_{\left[ {a,b} \right]}\phi:=\top U_{\left[ {a,b} \right]}\phi$ (eventually operator) and $G_{\left[ {a,b} \right]}\phi:=\neg F_{\left[ {a,b} \right]}\neg\phi$ (always operator). 
Let $(\mathbf{x},t)\models\phi$ denote the satisfaction relation. A formula $\phi$ is satisfiable if $\exists\mathbf{x}:\mathbb{R}_{\geq 0}\to\mathbb{R}^d$ such that $(\mathbf{x},t)\models\phi$.
\begin{definition}\cite{maler2004monitoring}
(STL Semantics): For a signal $\mathbf{x}:\mathbb{R}_{\geq 0}\to\mathbb{R}^d$, the STL semantics are recursively given by:
\begin{align*}
   &(\mathbf{x},t)\models\nu \;\;\;\;\;\;\;\;\;\;\;\;\;\Leftrightarrow h(\mathbf{x})\geq 0,\\
   &(\mathbf{x},t)\models\neg\phi\;\;\;\;\;\;\;\;\;\;\;\Leftrightarrow \neg((\mathbf{x},t)\models\phi),\\
   &(\mathbf{x},t)\models\phi'\wedge\phi''\;\;\;\;\Leftrightarrow (\mathbf{x},t)\models\phi'\wedge(\mathbf{x},t)\models\phi'',\\
    &(\mathbf{x},t)\models\phi' U_{\left[ {a,b} \right]}{\phi ''}\Leftrightarrow \exists t_1\in{\left[ {t+a,t+b} \right]}\; s.t. (\mathbf{x},t_1)\models\phi''\\&\;\;\;\;\;\;\;\;\;\;\;\;\;\;\;\;\;\;\;\;\;\;\;\;\;\;\;\;\;\;\;\;\;\wedge\forall t_2\in{\left[ {t,t_1} \right]}, (\mathbf{x},t_2)\models\phi',\\
    &(\mathbf{x},t)\models F_{\left[ {a,b} \right]}{\phi}\;\;\;\;\;\Leftrightarrow\exists t_1\in{\left[ {t+a,t+b} \right]}\; s.t. (\mathbf{x},t_1)\models\phi,\\
        &(\mathbf{x},t)\models G_{\left[ {a,b} \right]}{\phi}\;\;\;\;\;\Leftrightarrow\forall t_1\in{\left[ {t+a,t+b} \right]}\; s.t. (\mathbf{x},t_1)\models\phi. 
\end{align*}
\end{definition}
\subsection{Leader-follower multi-agent systems}\label{LF_form}
Consider a connected undirected graph $\mathcal{G}:=(\mathcal{V}, \mathcal{E})$, where $\mathcal{V}:=\{{1,\cdots,n}\}$ indicates the set consisting of $n$ agents and $\mathcal{E}\in \mathcal{V}\times\mathcal{V}$ represents communication links. 
Without loss of generality, we suppose the first $n_f$ agents as followers and the last $n_l$ agents as leaders, with corresponding vertices, sets denoted as $\mathcal{V}_f:=\{{1,\cdots,n_f}\}$ and $\mathcal{V}_l:=\{{n_f+1,\cdots,n_f+n_l}\}$, respectively, with $n_f+n_l=n$. 
Let $p_i\in\mathbb{R}$, $v_i\in\mathbb{R}$ and $u_i\in\mathbb{R}$ denote the position, velocity and control input of agent $i\in\mathcal{V}$, respectively.  Moreover, $\mathcal{N}_i$ denotes the set of neighbors of agent $i$ and $|\mathcal{N}_i|$ determines the cardinality of the set $\mathcal{N}_i$.
In addition, we define the stacked vector of all elements in the set $\mathcal{X}$ with cardinality $|\mathcal{X}|$, as $[x_i]_{i\in\mathcal{X}} :=[x_{i_1}^\top,\cdots, x_{i_{|\mathcal{X}|}}^\top]^\top$, $i_1,\cdots,i_{|\mathcal{X}|}\in\mathcal{X}$.
Then, the $1^{st}$ order dynamics of agent $i$ can be described as 
\begin{align}\label{agent_singl}
&\dot p_i=\mathfrak{f}_i^s(p_i,[p_j]_{j\in\mathcal{N}_i})+b_i \mathfrak{g}_i^s(p_i)u_i,
\end{align}
where $b_i=0$, $i\in\{{1,\cdots,n_f}\}$, indicates the followers and $b_i=1$, $i\in\{{n_f+1,\cdots,n_f+n_l}\}$, denotes the leaders.
 In addition,  $\mathfrak{f}_i^s:\mathbb{R}^{1+|\mathcal{N}_i|}\to\mathbb{R}$, $\mathfrak{g}_i^s:\mathbb{R}\to\mathbb{R}$ are assumed to be locally Lipschitz continuous functions. 
 
We also introduce the $2^{nd}$ order dynamics for the followers for $b_i=0$ and the leaders for $b_i=1$ as follows.
\begin{align}\label{agent_l}
&\dot p_i=v_i\nonumber\\
&\dot v_i=\mathfrak{f}_i^d(p_i, [p_j]_{j\in\mathcal{N}_i},v_i, [v_j]_{j\in\mathcal{N}_i})+b_i \mathfrak{g}_i^d(v_i)u_i,
\end{align}
 in which $\mathfrak{f}_i^d:\mathbb{R}^{2+2|\mathcal{N}_i|}\to\mathbb{R}$, $\mathfrak{g}_i^d:\mathbb{R}\to\mathbb{R}$ are locally Lipschitz continuous functions. 

We consider the STL fragment
\begin{subequations}
	\begin{align}
&\psi::= \top|\nu|\psi'\wedge\psi'',\label{1st}\\
&\phi::=G_{\left[ {a,b} \right]}\psi|F_{\left[ {a,b} \right]}\psi|\psi'U_{\left[ {a,b} \right]}\psi''|\phi'\wedge\phi'',\label{2nd}
\end{align}
\end{subequations}
where $\psi',\psi''$ are formulas of class $\psi$ in \eqref{1st} and $\phi',\phi''$ are formulas of class $\phi$ in \eqref{2nd}. 
It is worth mentioning that these formulas can be extended to consider disjunctions ($\vee$) using automata based approaches \cite{lindemann2020efficient}.

Consider formulas $\phi^s$ and $\phi^d$ of the form \eqref{2nd}, corresponding to the $1^{st}$ and $2^{nd}$ order leader-follower multi-agent systems, respectively. The formula $\phi^s$ (resp. $\phi^d$) consists of a number of temporal operators and its satisfaction depends on the behavior of the set of agents $\mathcal{V}=\{1,\cdots,n\}$.
By behavior of an agent $i$, we mean the state trajectories that evolve according to \eqref{agent_singl} (resp. \eqref{agent_l}).

\begin{assumption}\label{concave}
	Predicate functions in $\phi^s$ (resp. $\phi^d$) are concave.
	\end{assumption}
	Concave predicate functions contain linear functions as well as functions corresponding to reachability tasks ($\left\| x-p \right\|^2\leq \epsilon$, $p\in\R^n$, $\epsilon\geq 0$). As the minimum of concave predicate functions is again concave, 
	 they are useful in constructing valid control barrier functions \cite[Lemmas 3, 4]{lindemann2020barrier}.


	Based on \eqref{agent_singl} and \eqref{agent_l}, we write the stacked dynamics for the set of agents in $i\in\mathcal{V}$, as
	\begin{align}\label{agent_f1e}
{\dot x^s}=\mathfrak{f}^s(x^s)+\mathfrak{g}^s(x^s)u,
\end{align}
for the $1^{st}$ order dynamics and 
\begin{align}\label{agent_fe}
{\dot x^d}=\mathfrak{f}^d(x^d)+\mathfrak{g}^d(x^d)u,
\end{align}
for the $2^{nd}$ order dynamics,
where 
$x^s\!:=\left[x_{i}^s\right]_{i\in\mathcal{V}}=\left[p_i\right]_{i\in\mathcal{V}}\!\in\!\mathcal{S}^s\subseteq{\mathbb{R}^{n}}$,
$\mathfrak{f}^s(\cdot)=\left[\mathfrak{f}_{i}^s(\cdot)\right]_{i\in\mathcal{V}}\in\!\mathbb{R}^{n}$, $x^d:=\left[x_{i}^d\right]_{i\in\mathcal{V}}=\left[p_i;v_i\right]_{i\in\mathcal{V}}\in\mathcal{S}^d\subseteq{\mathbb{R}^{2n}}$,
$\mathfrak{f}^d(\cdot)=\left[\mathfrak{f}_{i}^d(\cdot)\right]_{i\in\mathcal{V}}\in\!\mathbb{R}^{2n}$.
 Without loss of generality, we consider functions $\mathfrak{f}_{i}^s(x^s)$ and $\mathfrak{f}_{i}^d(x^d)$ as $\mathfrak{f}_{i}^s(x^s)=\mathfrak{f}_{i,i}^s(x_{i}^s)+\sum\nolimits_{j\in\mathcal{V}, j\ne i}\mathfrak{f}_{i,j}^s(x_{i}^s,x_{j}^s)$and $\mathfrak{f}_{i}^d(x^d)=\mathfrak{f}_{i,i}^d(x_{i}^d)+	\sum\nolimits_{j\in\mathcal{V}, j\ne i}\mathfrak{f}_{i,j}^d(x_{i}^d,x_{j}^d)$, respectively. The local dynamic function $\mathfrak{f}_{i,i}^s(x_{i}^s)$ corresponds to the terms of $\mathfrak{f}_{i}^s(x^s)$ which are only dependent on $p_i^s$, and $\mathfrak{f}_{i,j}^s(x_{i}^s,x_{j}^s)$ contains the terms of $\mathfrak{f}_{i}^s(x^s)$ which are dependent on agent $j\in\mathcal{V}, j\ne i$ as well. The same holds for $\mathfrak{f}_{i}^d(x^d)$. We assume local dynamics of the agents are stable.
 For the case of one leader, with follower and leader sets $\mathcal{V}_{f}:=\{{1,\cdots,n-1}\}$ and $\mathcal{V}_{l}:=\{{n}\}$, respectively, the input matrices and control input signal are defined as $\mathfrak{g}^s(\cdot):=\left[ {\begin{array}{*{20}{c}}
0_{{n-1}\times 1}^T,\mathfrak{g}_{n}^s(\cdot)
\end{array}} \right]^T$, $\mathfrak{g}^d(\cdot):=\left[ {\begin{array}{*{20}{c}}
0_{{2n-1}\times 1}^T,\mathfrak{g}_{{n}}^d(\cdot)
\end{array}} \right]^T$, and $u:=u_{n}\in\mathbb{R}$.
Note that the input matrices $\mathfrak{g}^s(\cdot)$ and $\mathfrak{g}^d(\cdot)$ are not full row rank. 

\begin{definition} \cite{ames2016control}
      A continuous function $\lambda:(-b,a)\Rightarrow\mathbb{R}$ for some $a,b>0$ is called an extended class $\mathcal{K}$ function if it is strictly increasing and $\lambda(0)=0$.  
\end{definition}
			\subsection{Time-varying barrier functions}\label{solution}
			Let $\mathfrak{h}^s( x^s,t):\mathbb{R}^{ n}\times\mathbb{R}_{\geq 0}\to\mathbb{R}$ (resp. $\mathfrak{h}^d( x^d,t):\mathbb{R}^{ 2n}\times\mathbb{R}_{\geq 0}\to\mathbb{R}$) be a piece-wise differentiable function.
			The time-varying barrier function $\mathfrak{h}^s(x^s,t)$ (resp. $\mathfrak{h}^d(x^d,t)$) is built corresponding to the STL task $\phi^s$ (resp. $\phi^d$) related to the multi-agent system \eqref{agent_f1e} (resp. \eqref{agent_fe}).
				 Consider the $1^{st}$ order dynamic network \eqref{agent_f1e}. 
Following the procedure in \cite{lindemann2020barrier}, we construct the barrier function, piece-wise continuous in the second argument, for the conjunctions of a number of $q^s$ single temporal operators in $\phi^s$, by using a smooth under-approximation of the min-operator. Accordingly, consider the continuously differentiable barrier functions ${\mathfrak{h}_j^s}(x^s,t)$, ${j \in \{ 1, \cdots ,{q^s}\} }$, corresponding to each temporal operator in $\phi^s$. Then, we have $\mathop {\min }\limits_{j \in \{ 1, \cdots ,{q^s}\} } {\mathfrak{h}_j^s}({{x}^s},t) \approx -\frac{1}{\eta^s}\rm{ln}(\sum\limits_{\it{j} = 1}^{\it{q^s}} {\exp ( - {\it{\eta^s{\mathfrak{h}_j^s}}}({\it{{x}^s}},t))} )$, with parameter $\eta^s>0$ that is proportionally related to the accuracy of this approximation.
	In view of~\cite[Steps A, B, and C]{lindemann2020barrier}, the corresponding barrier function to $\phi^s$ could be constructed as 
		\begin{align}\label{barrier}
		   \mathfrak{h}^s( x^s,t):=-\frac{1}{\eta^s}\rm{ln}(\sum\limits_{\it{j} = 1}^{\it{q^s}} {\exp ( - {\it{\eta}^s{{\mathfrak{h}_j^s}}}({\it{{ x}^s}},t))}),
		\end{align}
		 where each ${\it{{\mathfrak{h}_j^s}}}({\it{{ x}^s}},t)$ is related to an always or eventually operator specified for the time interval $\left[a_{j},b_{j}\right]$. Whenever the $j$th temporal operator is satisfied, its corresponding barrier function ${\it{{\mathfrak{h}_j^s}}}({\it{{x}^s}},t)$ is deactivated and hence a switching occurs in $ \mathfrak{h}^s( x^s,t)$. This time-varying strategy helps reducing the conservatism in the presence of large numbers of conjunctions \cite{lindemann2020barrier}. 
		Due to the knowledge of $\left[a_{j},b_{j}\right]$, the switching instants can be known in advance. Denote the switching sequence as $\{\tau_{0}:=t_0,\tau_{1},\cdots,\tau_{p^s}\}$.
		At time $t\geq \tau_{l}$, the next switch occurs at $\tau_{l+1}:=\rm{argmin}\it{_{b_{j}\in\{b_{1},...,b_{q^s}\}}\zeta(b_{j},t)}$, ${l\in\{0,\cdots,p^s-1\}}$, where $\zeta(b_{j},t):=\left\{ \begin{array}{l}
b_{j} - t,\;\;b_{j} - t > 0\\
\infty,\;\;\;\;\;\;\;\;\rm{otherwise}
\end{array} \right.$.
		\begin{definition}(Forward Invariance)
			Consider the set
		\begin{align}\label{safe_set}
		\mathfrak{C}^s(t):=\{x^s\in\mathbb{R}^{n}|\mathfrak{h}^s( x^s,t)\geq 0\}.
		\end{align} 
		The set $\mathfrak{C}^s(t)$ is forward invariant with a given control law $u$ for \eqref{agent_f1e}, if for each initial condition $x_0^s\in \mathfrak{C}^s(t_0)$, there exists a unique solution $x^s : [t_0, t_1] \to \mathbb{R}^n$ with $x(t_0) = x_0^s$, such that $x^s(t)\in \mathfrak{C}^s(t)$ for all $t\in [t_0, t_1]$.
		\end{definition}
		If $\mathfrak{C}^s(t)$
		is forward invariant, then it holds that $x^s\models\phi^s$.
			Note that since at each switching instant, one control barrier function ${\mathfrak{h}_j^s}({\it{{ x}^s}},t))$ is removed from $ \mathfrak{h}^s( x^s,t):=-\frac{1}{\eta^s}\rm{ln}(\sum\limits_{\it{j} = 1}^{\it{q^s}} {\exp ( - {\it{\eta}^s{{\mathfrak{h}_j^s}}}({\it{{ x}^s}},t))})$, the set $\mathfrak{C}^s(t)$ is non-decreasing at these switching instants.
		Hence, for each switching instant $\tau_{l}$, it holds that $\mathop {\lim }\limits_{t  \to \tau_{l}^{-} } {\mathfrak{C}^s}(t ) \subseteq {\mathfrak{C}^s}(\tau_{l})$, where $\mathop {\lim }\limits_{t  \to \tau_{l}^{-} } {\mathfrak{C}^s}(t)$ is the left-sided limit of ${\mathfrak{C}^s}({t})$ at $t=\tau_{l}$. 
	
			We also assume that
		the set $\mathfrak{C}^s$ is compact and non-empty.
\begin{definition}
     We denote the set ${\mathfrak{C}^s}(t)$ to be \textbf{fixed-time convergent} for \eqref{agent_f1e}, if there exists a  user-defined, independent of the initial condition, and finite time $T^s>t_0$, such that $\lim_{t\to T^s}x^s(t)\in{\mathfrak{C}^s}(t)$. Moreover, the set ${\mathfrak{C}^s}(t)$ is \textbf{robust fixed-time convergent} if $\lim_{t\to T^s}x^s(t)\in{\mathfrak{C}_{rf}^s}(t)$, where ${\mathfrak{C}_{rf}^s}(t)\supset{\mathfrak{C}^s}(t)$, and \textbf{robust convergent} for \eqref{agent_f1e}, if $\lim_{t\to \infty}x^s(t)\in{\mathfrak{C}_{rf}^s}(t)$.
     The set ${\mathfrak{C}_{rf}^s}(t)$ is characterized as $\mathfrak{C}_{rf}^s(t):=\{x^s\in\mathbb{R}^{n}|\mathfrak{h}^s( x^s,t)\geq -\epsilon_{\max}^s\}$, where $\epsilon_{\max}^s$ is a bounded and positive value.
\end{definition}
		The same properties hold for the barrier functions $\mathfrak{h}^d(x^d,t)$ and the set $\mathfrak{C}^d(t)$ corresponding to the $2^{nd}$ order dynamic network \eqref{agent_fe} under the task $\phi^d$.
		
\section{First order leader-follower multi-agent systems}\label{solution3}
In this section, we provide conditions to guarantee the \emph{fixed-time convergence} property of the set $\mathfrak{C}^s(t)$ corresponding to the STL task of the form \eqref{2nd}, using control barrier certificates for a network of $1^{st}$ order leader-follower agents, based on the leader information of the involved followers.
		Consider the leader-follower network \eqref{agent_f1e} under the task $\phi^s$.
		Let $\mathfrak{h}^s(x^s,t)$ define a time-varying barrier function for this system. 
							 Next, we provide a Lemma to guarantee the \emph{fixed-time convergence and forward invariance} of the set $\mathfrak{C}^s(t)$ given in \eqref{safe_set} for system  \eqref{agent_f1e}, under the following assumption.
		 \begin{assumption}\label{star}
							     The leader agent corresponding to the graph $\mathcal{G}:=(\mathcal{V}, \mathcal{E})$ subject to the task $\phi^s$ has knowledge of the functions $ \frac{{\partial\mathfrak{h}^s({{ x}^s},t)}}{{\partial {x_i^s}}}$ and dynamics ${\mathfrak{f}_{i}^s({x^s})}$, ${i\in\left\{1,\cdots,n\right\}}$. 
							 \end{assumption}
							 A special case satisfying Assumption \ref{star}, is the star topology network with a leader in the middle.
				\begin{lemma}\label{lem1}
					Consider a leader-follower network subject to the dynamics \eqref{agent_f1e} containing one leader, under STL task $\phi^s$ of the form \eqref{2nd} satisfying Assumption \ref{concave}. Suppose that the leader satisfies Assumption \ref{star}.
					Let $\mathfrak{h}^s(x^s,t)$ be a time-varying barrier function associated with the task $\phi^s$, specified in Section \ref{solution}. 
					If for some open set $\mathcal{S}^s$ with $\mathcal{S}^s\supset \mathfrak{C}^s(t)$ , $\forall t\geq t_0$, and for all $(x^s,t)\in \mathcal{S}^s\times[\tau_{l},\tau_{l+1})$, ${l\in\{0,\cdots,p^s-1\}}$, for some constants   $0<\gamma_{1}^s<1$,
					$\gamma_{2}^s>1$, $\alpha^s>0$, $\beta^s>0$ such that $\frac{1}{\alpha^s(1-\gamma_{1}^s)}+\frac{1}{\beta^s(\gamma_{2}^s-1)}\leq\min_{l\in\{0,\cdots,p^s-1\}}\{\tau_{l+1}-\tau_{l}\}$, there exists a control law $u_{n}$ satisfying
							\begin{equation}\label{Ineq1}
					\begin{array}{l}
				\sum\nolimits_{i\in\mathcal{V}} {\frac{{\partial 	\mathfrak{h}^s({{ x}^s},t)}}{{\partial {x_{i}^s}}} {\mathfrak{f}_{i}^s({x^s})}  + \frac{{\partial 	\mathfrak{h}^s({{ x}^s},t)}}{{\partial {x_{n}^s}}}{\mathfrak{g}_{{n}}^s}({x_{n}^s}){u_{n}}}\\+ \frac{{\partial \mathfrak{h}^s({{ x}^s},t)}}{{\partial t}}\ge - {\alpha^s}\;\sgn({	\mathfrak{h}^s({{ x}^s},t)}){	|\mathfrak{h}^s({{ x}^s},t)|}^{\gamma _{1}^s}\\ \;\;\;\;\;\;\;\;\;\;\;\;\;\;\;\;\;\;\;\;-{\beta^s}\;\sgn({	\mathfrak{h}^s({{x}^s},t)}){	|\mathfrak{h}^s({{x}^s},t)|}^{\gamma _{2}^s},
					\end{array}
					\end{equation}
					then the set $\mathfrak{C}^s(t)$ is fixed-time convergent and forward invariant.
					Hence, $x^s\models\phi^s$.
				\end{lemma}
				\begin{proof}
				Consider the inequality \eqref{Ineq1} and dynamics \eqref{agent_f1e}.  Since the leader control signal $u_{n}$ is the only external input responsible for controlling the network, and under Assumption \ref{star}, the inequality \eqref{Ineq1} can be written as follows:
				\begin{align}\label{cond2}
		   &\frac{{\partial {\mathfrak{h}^s({{ x^s}},t)}}}{{\partial {x^s}}}\left( {{\mathfrak{f}^s}({x^s}) + {\mathfrak{g}^s}({x^s}){u}} \right)
+ \frac{{\partial {	\mathfrak{h}^s({x^s},t)}}}{{\partial t}}\nonumber\\& \;\;\;\;\;\;\;\;+ {\alpha^s}\emph{\sgn}({	\mathfrak{h}^s({x^s},t)}){|\mathfrak{h}^s({{x^s}},t)|}^{\gamma _{1}^s}\nonumber\\& \;\;\;\;\;\;\;\; +{\beta^s}\emph{\sgn}({	\mathfrak{h}^s({{x^s}},t)}){|\mathfrak{h}^s({{x^s}},t)|}^{\gamma _{2}^s}\geq 0.
		\end{align}
				Now, consider the satisfaction of \eqref{cond2} for all $(x^s,t)\in\mathcal{S}^s\times[\tau_{l},\tau_{l+1})$ 
				under a control input $u:= u_{n}\in\mathbb{R}$ with positive constants $\gamma_{1}^s<1$, $\gamma_{2}^s>1$, $\alpha^s$, $\beta^s$.
Note that by $\mathop {\lim }\limits_{t  \to \tau_{l}^{-} } {\mathfrak{C}^s}(t ) \subseteq {\mathfrak{C}^s}(\tau_{l})$, it is sufficient to ensure convergence and forward invariance of ${\mathfrak{C}^s}(t)$ for each $[\tau_{l},\tau_{l+1})$.
This is due to the fact that if $\mathfrak{h}^s(x^s,t)\in\mathfrak{C}^s(t)$ for all $t\in[\tau_{l},\tau_{l+1})$, then $\mathfrak{h}^s(x^s,\tau_{l+1})\in\mathfrak{C}^s(\tau_{l+1})$.
		Consider the function $V^s(x^s,t)=\max {\{ 0,-\mathfrak{h}^s(x^s,t)}\}$. Then, for $x^s(t_0)\in \mathfrak{C}^s(t_0)$ ($\mathfrak{h}^s(x^s,t)\geq 0$) we have $V^s(x^s,t)=0$ for all $t\geq t_0$ by the Comparison Lemma \cite{bhat2000finite}. Hence, the set $\mathfrak{C}^s(t)$ is forward-invariant.
	Moreover, for $x^s(t_0)\in \mathcal{S}^s\backslash \mathfrak{C}^s(t_0)$ ($\mathfrak{h}^s(x^s,t)< 0$), we get $V^s(x^s,t)=-\mathfrak{h}^s(x^s,t)$. Thus, \eqref{cond2} can be written as
	\begin{equation*}
	\dot V^s(x^s,t)\leq - {\alpha^s}	{V^s(x^s,t)}^{\gamma _{1}^s} -{\beta^s} {V^s(x^s,t)}^{\gamma _{2}^s},
	\end{equation*}
	which guarantees the fixed-time convergence of $x^s$ to the set $\mathfrak{C}^s(t)$ within $T^s\leq\frac{1}{\alpha^s(1-\gamma_{1}^s)}+\frac{1}{\beta^s(\gamma_{2}^s-1)}$ and staying there onwards, according to \cite{bhat2000finite}. 
		 The proof is complete.
				\end{proof}

			Inspired by \cite[Theorem 2]{black2020quadratic}, we extend the results of Lemma \ref{lem1} to the case of leader partial information from the subgraph, i.e., there exist followers that aren't neighbors of the leader, denoted by $i\notin\mathcal{N}_{n}$. In this case, the \emph{robust fixed-time convergence} property of the set $\mathfrak{C}^s(t)$ is guaranteed.
			
				Next, we impose relaxations on Assumption \ref{star} and provide further results on task satisfaction under new conditions.
					\begin{assumption}\label{bound}
				    Consider the $1^{st}$ order leader-follower network \eqref{agent_f1e} with a single leader $i=n$. We assume that there exists a positive constant $\delta^s$ satisfying $\|\sum\nolimits_{i\in\mathcal{N}_{n}, j\notin\mathcal{N}_{n}} \frac{{\partial \mathfrak{h}^s({{ x}^s},t)}}{{\partial {x_{j}^s}}} {\mathfrak{f}_j^s({x^s})}+\frac{{\partial \mathfrak{h}^s({{ x}^s},t)}}{{\partial {x_{i}^s}}} {\mathfrak{f}_{i,j}^s({x_{i}^s,x_{j}^s})}\|\leq \delta^s$, $\forall(x^s,t)\in\mathcal{S}^s\times[\tau_{l},\tau_{l+1})$, ${l\in\{0,\cdots,p^s-1\}}$. 
				\end{assumption}
				\begin{remark}
				Note that the function $\mathfrak{h}^s({{x}^s},t)$ is differentiable $\forall(x^s,t)\in\mathcal{S}^s\times[\tau_{l},\tau_{l+1})$. Moreover, $\mathfrak{f}^s(x^s)$ is Lipschitz and subject to a connected network.
				Hence, by stability of local state dynamics, we can argue about the boundedness of the stack vector $x^s$. Thus, Assumption \ref{bound} is not strong. Moreover, there is no necessity for the leader to know $\delta^s$. This term is used in determining the ultimate convergent set, as will be demonstrated in the following theorem.
				\end{remark}
					\begin{theorem}\label{lem11}
					Consider a leader-follower multi-agent network subject to the dynamics \eqref{agent_f1e} containing one leader, under STL task $\phi^s$ of the form \eqref{2nd} satisfying Assumption \ref{concave}. 
					Let $\mathfrak{h}^s(x^s,t)$ be a time-varying barrier function associated with the task $\phi^s$, specified in Section \ref{solution}. 
					Suppose that Assumption \ref{bound} is satisfied for the network \eqref{agent_f1e}.
					If for some constants  $\mu^s>1$, $k^s>1$, $\gamma_{1}^s=1-\frac{1}{\mu^s}$, $\gamma_{2}=1+\frac{1}{\mu^s}$, $\alpha^s>0$, $\beta^s>0$, for some open set $\mathcal{S}^s$ with $\mathcal{S}^s\supset \mathfrak{C}^s(t)$, $\forall t\geq 0$, and for all $(x^s,t)\in \mathcal{S}^s\times[\tau_{l},\tau_{l+1})$, $l\in\{0,\cdots,p^s-1\}$, there exists a control law $u_{n}$ such that
							\begin{equation}\label{Ineq2}
					\begin{array}{l}
				\sum\nolimits_{i\in\mathcal{N}_{n}}\frac{{\partial 	\mathfrak{h}^s({{ x}^s},t)}}{{\partial {x_{i}^s}}} {\mathfrak{f}_{i,i}^s({x_{i}^s})}+\frac{{\partial \mathfrak{h}^s({{ x}^s},t)}}{{\partial {x_{n}^s}}} {\mathfrak{f}_{n,i}^s({x_{n}^s,x_{i}^s})}\\+\frac{{\partial \mathfrak{h}_e^s({{ x}^s},t)}}{{\partial {x_{n}^s}}} {\mathfrak{f}_{n,n}^s({x_{n}^s})} + \frac{{\partial \mathfrak{h}^s({{ x}^s},t)}}{{\partial t}}+ \frac{{\partial \mathfrak{h}^s({{ x}^s},t)}}{{\partial {x_{n}^s}}}{\mathfrak{g}_{n}^s}({x_{n}^s}){u_{n}}\\\ge - {\alpha^s}\;\sgn({\mathfrak{h}^s({{ x}^s},t)}){	|\mathfrak{h}^s({{ x}^s},t)|}^{\gamma _{1}}\\ -{\beta^s}\;\sgn({	\mathfrak{h}^s({{x}^s},t)}){	|\mathfrak{h}^s({{x}^s},t)|}^{\gamma _{2}^s},
					\end{array}
					\end{equation}
					with
						\begin{align}\label{T}
						T^s &\le \left\{ \begin{array}{l}
							\frac{\mu^s }{{{\alpha^s}(c^s - b^s)}}\log (\frac{{\left|1+ c^s \right|}}{{\left|1+ b^s \right|}})\;\;\;\;;{\delta^s} > 2\sqrt {{\alpha^s}{\beta^s}} \\
							\frac{\mu^s }{{\sqrt {{\alpha^s}{\beta^s}} }}(\frac{1}{{ k^s-1}})\;\;\;\;\;\;\;\;\;\;\;\;\;\;\;;{\delta^s} = 2\sqrt {{\alpha^s}{\beta^s}} \\
							\frac{\mu^s }{{{\alpha^s}{k_{1}^s}}}(\frac{\pi }{2} - {\tan ^{ - 1}}{k_{2}^s})\;\;\;\;;0\leq{\delta^s} < 2\sqrt {{\alpha^s}{\beta^s}}
						\end{array} \right.\nonumber\\&\leq \min_{l\in\{0,\cdots,p^s-1\}}\{\tau_{l+1}-\tau_{l}\},
					\end{align}
					where $b^s, c^s$ are the solutions of $\gamma^s(s)={\alpha^s}s^2-{\delta^s}s+\beta^s=0$, $k_{1}^s=\sqrt{\frac{4\alpha^s\beta^s -{\delta^s}^2}{4{\alpha^s}^2}}$, $k_{2}^s=-\frac{\delta^s}{\sqrt{{4\alpha^s\beta^s-{\delta^s}^2}}}$, and $\delta^s$ is introduced in Assumption \ref{bound}, then, the set $\mathfrak{C}_{{rf}}^s(t)\supset \mathfrak{C}^s(t)$ defined by
						\begin{align*}
		\mathfrak{C}_{{rf}}^s(t):=\{ x^s\in\mathbb{R}^{n}|\mathfrak{h}^s( x^s,t)\geq -\epsilon_{\max}^s\}
		\end{align*}
		with
			\begin{align}\label{D3}
\epsilon_{\max}^s= \left\{ \begin{array}{l}
							  {(\frac{{{\delta^s} + \sqrt {{\delta^s}^2 - 4{\alpha^s}{\beta^s}} }}{{2{\alpha^s}}})}^{\mu^s}\;\;\;;{\delta^s} > 2\sqrt {{\alpha^s}{\beta^s}} \\
							{{k^s}^{\mu^s} }{{(\frac{{{\beta^s}}}{{{\alpha^s}}})}^{\frac{\mu^s }{2}}}\;\;\;\;\;\;\;\;\;\;\;\;\;\;\;;{\delta^s} = 2\sqrt {{\alpha^s}{\beta^s}} \\
							{{\frac{{{\delta^s} }}{{2{\sqrt {{\alpha^s}{\beta^s}} }}}}}\;\;\;\;\;\;\;\;\;\;\;\;\;\;\;\;\;\;\;\;\;;0\leq{\delta^s} < 2\sqrt {{\alpha^s}{\beta^s}},
						\end{array} \right.
					\end{align}
		is forward invariant and fixed-time convergent within $T^s$ time units, defined in \eqref{T}.
				\end{theorem}
				\begin{proof}
				Inequality \eqref{Ineq2} can be written as 
					\begin{equation}\label{Ineq3}
					\begin{array}{l}
				\sum\nolimits_{i\in\mathcal{N}_{n}} \frac{{\partial \mathfrak{h}^s({{ x}^s},t)}}{{\partial {x_{i}^s}}} {\mathfrak{f}_{i,i}^s({x_{i}^s})}+\frac{{\partial \mathfrak{h}^s({{ x}^s},t)}}{{\partial {x_{n}^s}}} {\mathfrak{f}_{n,i}^s({x_{n}^s,x_{i}^s})} \\+ \sum\nolimits_{i\in\mathcal{N}_{n}, j\notin\mathcal{N}_{n}} \frac{{\partial \mathfrak{h}^s({{ x}^s},t)}}{{\partial {x_{j}^s}}} {\mathfrak{f}_j^s({x^s})}+\frac{{\partial \mathfrak{h}^s({{ x}^s},t)}}{{\partial {x_{i}^s}}} {\mathfrak{f}_{i,j}^s({x_{i}^s,x_{j}^s})}\\ + \frac{{\partial \mathfrak{h}^s({{ x}^s},t)}}{{\partial {x_{n}^s}}} {\mathfrak{f}_{n,n}^s({x_{n}^s})}+\frac{{\partial 	\mathfrak{h}^s({{ x}^s},t)}}{{\partial {x_{n}^s}}}{\mathfrak{g}_{n}^s}({x_{n}^s}){u_{n}}+ \frac{{\partial \mathfrak{h}^s({{ x}^s},t)}}{{\partial t}}\\\ge - {\alpha^s}\;\emph{\sgn}({	\mathfrak{h}^s({{ x}^s},t)}){	|\mathfrak{h}^s({{ x}^s},t)|}^{\gamma _{1}}\\ -{\beta^s}\;\emph{\sgn}({	\mathfrak{h}^s({{x}^s},t)}){	|\mathfrak{h}^s({{x}^s},t)|}^{\gamma _{2}^s}\\
				+ \sum\nolimits_{i\in\mathcal{N}_{n}, j\notin\mathcal{N}_{n}} \frac{{\partial \mathfrak{h}^s({{ x}^s},t)}}{{\partial {x_{j}^s}}} {\mathfrak{f}_j^s({x^s})}+\frac{{\partial \mathfrak{h}^s({{ x}^s},t)}}{{\partial {x_{i}^s}}} {\mathfrak{f}_{i,j}^s({x_{i}^s,x_{j}^s})}.
					\end{array}
					\end{equation}
					Then, we get
								\begin{equation}\label{Ineq3}
					\begin{array}{l}
				\frac{{\partial \mathfrak{h}^s({{ x}^s},t)}}{{\partial {x^s}}} ({\mathfrak{f}^s({x^s})}+ {\mathfrak{g}_{n}^s}({x_{n}^s}){u_{n}})+ \frac{{\partial \mathfrak{h}^s({{ x}^s},t)}}{{\partial t}}\\\ge - {\alpha^s}\;\emph{\sgn}({	\mathfrak{h}^s({{ x}^s},t)}){	|\mathfrak{h}^s({{ x}^s},t)|}^{\gamma _{1}}\\ -{\beta^s}\;\emph{\sgn}({	\mathfrak{h}^s({{x}^s},t)}){	|\mathfrak{h}^s({{x}^s},t)|}^{\gamma _{2}^s}\\
				+ \sum\nolimits_{i\in\mathcal{N}_{n}, j\notin\mathcal{N}_{n}} \frac{{\partial \mathfrak{h}^s({{ x}^s},t)}}{{\partial {x_{j}^s}}} {\mathfrak{f}_j^s({x^s})}+\frac{{\partial \mathfrak{h}^s({{ x}^s},t)}}{{\partial {x_{i}^s}}} {\mathfrak{f}_{i,j}^s({x_{i}^s,x_{j}^s})}.
					\end{array}
					\end{equation}
				It is apparent that the left hand side of \eqref{Ineq3} is equal to the one in \eqref{cond2}, since ${\mathfrak{g}_{n}^s}({x_{n}^s}){u_{n}}={\mathfrak{g}^s}({x^s}){u}$.
				Following the proof of Lemma \ref{lem1}, function $V^s(x^s,t)=\max {\{ 0,-\mathfrak{h}^s(x^s,t)}\}$ is considered. This function satisfies $V^s(x^s,t)=0$ for $x^s(t_0)\in \mathfrak{C}^s(t_0)$.  Therefore, as long as $\mathfrak{h}^s(x^s,t)\geq 0$, $V^s$ remains $0$ and then $ x^s(t)\in \mathfrak{C}^s(t)$, $t\geq t_0$. This ensures the forward invariance of $\mathfrak{C}^s(t)$.
	Moreover, $V^s(x^s,t)>0$ for $x^s\in \mathcal{S}^s\backslash \mathfrak{C}^s(t)$ and by Assumption \ref{bound}, \eqref{Ineq3} can be written as 
					\begin{equation*}
					\begin{array}{l}
				\dot V^s(x^s,t)\leq - {\alpha^s}	{V^s(x^s,t)}^{\gamma _{1}^s} -{\beta^s} {V^s(x^s,t)}^{\gamma _{2}^s}+\delta^s.
					\end{array}
					\end{equation*}
	Thus, according to \cite[Lemma 1]{sharifi2021}, the convergence of $V^s(x^s,t)$ to the set $\mathfrak{C}_{rf}^s(t)\supset \mathfrak{C}^s(t)$ in a fixed-time interval $t\leq T^s$, as in \eqref{T}, is achieved. In addition, considering the forward-invariance of $\mathfrak{C}^s(t)$ besides the convergence property of $\mathfrak{C}_{rf}^s(t)$, ensures forward-invariance of $\mathfrak{C}_{rf}^s(t)$
				\end{proof}
\begin{remark}
Note that due to lack of full information of the leader from the followers, a violation in the constraints satisfaction for $\phi^s$ might occur. This violation has been quantified as a function of $\delta^s$, demonstrated in \eqref{D3}. 
Furthermore, \eqref{T} is feasible provided that the minimum time interval between successive switchings is sufficiently large, such that user defined constants $\alpha^s$, $\beta^s$, $\mu^s$, $k^s$ fulfill \eqref{Ineq2} and \eqref{T}. 
\end{remark}

\section{Higher order leader-follower multi-agent systems}\label{solution2}
In this section, we consider higher-order dynamics multi-agent systems and in order to tackle higher relative degree specifications, provide a class of higher order control barrier functions with the property of convergence to the desired sets and robustness with respect to uncertainties.
\subsection{Convergent higher order control barrier functions}\label{Higher_B}
	Consider the autonomous system 
	\begin{align}\label{BF}
	   \dot{\mathbf{x}}=f(\mathbf{x}), 
	\end{align}
	with $\mathbf{x}\in\mathbb{R}^n$ and locally Lipschitz continuous function $f:\mathbb{R}^n\to\mathbb{R}^n$.
			We introduce class $C^m$ functions $\mathfrak{h}(
	\mathbf{x},t):\mathbb{R}^{n}\times\left[t_0,\infty\right)\to\mathbb{R}$, later called time-varying convergent higher order control barrier functions, to satisfy STL task $\phi$ of the form \eqref{2nd}.
			Define a series of functions $\psi_k:\mathbb{R}^{n}\times \left[{t_0,\infty}\right)\to\mathbb{R}^{n}$, $0\leq k\leq m$, as
			\begin{align}\label{sets}
			&\psi_0(\mathbf{x},t) := \mathfrak{h}(\mathbf{x},t),\nonumber\\
			&\psi_k(\mathbf{x},t) := \dot \psi_{k-1}(\mathbf{x},t)\nonumber\\&\;\;\;\;\;\;\;\;\;\;\;\;\;\;\;\;\; + \lambda_k(\psi_{k-1}(\mathbf{x},t)),\;1\leq k\leq m-1,\nonumber\\
			&\psi_m(\mathbf{x},t) := \dot \psi_{m-1}(\mathbf{x},t) \nonumber\\&\;\;\;\;\;\;\;\;\;\;\;\;\;\;\;\;\;+ \alpha_m\emph{\sgn}({\psi_{m-1}(\mathbf{x},t)})|\psi_{m-1}(\mathbf{x},t)|^{\gamma _{1m}}\nonumber\\&\;\;\;\;\;\;\;\;\;\;\;\;\;\;\;\;\;+\beta_m\emph{\sgn}({\psi_{m-1}(\mathbf{x},t)})|\psi_{m-1}(\mathbf{x},t)|^{\gamma _{2m}},
			\end{align}
			where $\lambda_k(\cdot)$, $k=1,\cdots, m-1$,  are $(m-k)^{th}$\red{-} order differentiable extended class $\mathcal{K}$ functions 
			and $0<\gamma _{1m}<1, \gamma _{2m}>1$, $\alpha_m>0, \beta_m>0$, are user specified constants.
		We define a series of sets $\mathfrak{C}_{k}(t)$, $k=1,\cdots,m$, assumed to be compact, as
		\begin{align}\label{safe_sets2}
		&\mathfrak{C}_k(t):=\{\mathbf{x}\in\mathbb{R}^{n}|\psi_{k-1}(\mathbf{x},t)\geq 0\}.
		\end{align}  
		\begin{definition}\label{HOB}
		 A class $C^m$ function $\mathfrak{h}(\mathbf{x},t):\mathbb{R}^{n}\times\left[t_0,\infty\right)\to\mathbb{R}$ is a time-varying convergent higher order barrier function (TCHBF) of degree $m$ for the system \eqref{BF}, if there exist extended class $\mathcal{K} $ functions $\lambda_k(\cdot)$, $k=1,\cdots, m-1$, constants $0<\gamma_{1m}<1, \gamma _{2m}>1, \alpha_m>0, \beta_m>0$, and an open set $\mathfrak{D}$ with $\mathfrak{C}:=\cap_{k=0}^{m} \mathfrak{C}_{k}\subset\mathfrak{D}\subset\mathbb{R}^{n}$ such that
		\begin{align*}
		    \psi_{m}(\mathbf{x},t)\geq 0,\;\forall (\mathbf{x},t)\in\mathfrak{D}\times\mathbb{R}_{\geq 0},
		\end{align*}
		where $\psi_k(\mathbf x,t)$, $k=0,\cdots, m$, are given in \eqref{sets}.
		\end{definition}
		Next, we aim to show the \emph{convergence and forward invariance} of the set $\mathfrak{C}$. 
		\begin{proposition}\label{prop1}
		    The set $\mathfrak{C}:=\cap_{k=1}^{m} \mathfrak{C}_{k}\subset\mathfrak{D}\subset\mathbb{R}^{n}$ is convergent and forward invariant for system \eqref{BF}, if $\mathfrak{h}(\mathbf{x},t)$ is a TCHBF.
		\end{proposition}
		\begin{proof}
			First, we show forward invariance of the set $\mathfrak{C}$.
		If $\mathfrak{h}(\mathbf{x},t)$ is a TCHBF, then $ \psi_{m}(\mathbf{x},t)\geq 0,\;\forall (\mathbf{x},t)\in\mathfrak{D}\times\left[t_0,\infty\right)$ according to Definition \ref{HOB}.
		Then,
		\begin{align*}
		  \dot \psi_{m-1}(\mathbf{x},t)& + \alpha_m\emph{\sgn}({\psi_{m-1}(\mathbf{x},t)})|\psi_{m-1}(\mathbf{x},t)|^{\gamma _{1m}}\\&+\beta_m\emph{\sgn}({\psi_{m-1}(\mathbf{x},t)})|\psi_{m-1}(\mathbf{x},t)|^{\gamma _{2m}}\geq 0.  
		\end{align*}
	By the proof of Lemma \ref{lem1}, it is concluded that if $\mathbf{x}(t_0)\in\mathfrak{C}_m(t_0)$, then we get $\psi_{m-1}(\mathbf{x},t)\geq 0,\;\forall t\in\left[t_0,\infty\right)$.
 Then, by \cite[Lemma 2]{glotfelter2017nonsmooth} and considering $\psi_{m-1}(\mathbf{x},t)$ given by \eqref{sets}, since $x(t_0)\in\mathfrak{C}_{m-1}(t_0)$, we also have $\psi_{m-2}(\mathbf{x},t)\geq 0,\;\forall t\in\left[t_0,\infty\right)$.
Iteratively, we can show $\psi_{k-1}(\mathbf{x},t)\geq 0,\;\forall t\in\left[t_0,\infty\right)$ for all $k\in\left\{{1,2,\cdots,m}\right\}$ which certifies $\mathbf{x}(t)\in\mathfrak{C}_k(t)$. Therefore, the set $\mathfrak{C}:=\cap_{k=1}^{m} \mathfrak{C}_{k}\subset\mathfrak{D}\subset\mathbb{R}^{n}$ is \emph{forward invariant}.
The proof of convergence property is similar to \cite[Proposition 3]{tan2021high} and is omitted due to space limitation.
		\end{proof}
			\begin{definition}\label{TFHOCBF}
		Consider the system 
		\begin{align}\label{cont_sys}
		    \dot {\mathbf{x}}={f}(\mathbf{x})+{g}(\mathbf{x})\mathbf{u},
		\end{align}
		with locally Lipschitz continuous functions $f$ and $g$.
		 A class $C^m$ function $\mathfrak{h}(\mathbf{x},t):\mathbb{R}^{ n}\times\left[t_0,\infty\right)\to\mathbb{R}$ is called a time-varying convergent higher order control barrier function (TCHCBF) of degree $m$ for this system under task $\phi$ of the form \eqref{2nd}, if there exist constants $0<\gamma_{1m}<1$, $\gamma_{2m}>1$, $\alpha_m>0$, $\beta_m>0$, and an open set $\mathfrak{D}$ with $\mathfrak{C}:=\cap_{k=1}^{m} \mathfrak{C}_{k}\subset\mathfrak{D}\subset\mathbb{R}^{n}$, $\mathfrak{C}_{k}$, $k=1,\cdots,m$, are defined as in \eqref{safe_sets2},  such that
		\begin{align}\label{cond}
		   &\frac{{\partial {\psi_{m-1}({{ \mathbf{x}}},t)}}}{{\partial {\mathbf{x}}}}\left( {{f}(\mathbf{x}) + {g}(\mathbf{x})\mathbf{u}} \right)
+ \frac{{\partial {	\psi_{m-1}(\mathbf{x},t)}}}{{\partial t}}\nonumber\\& \;\;\;\;\;\;\;\;\geq- {\alpha _m}\sgn({	\psi_{m-1}(\mathbf{x},t)}){|\psi_{m-1}(\mathbf{x},t)|}^{\gamma _{1m}}\nonumber\\& \;\;\;\;\;\;\;\; -{\beta _m}\sgn({	\psi_{m-1}({{\mathbf{x}}},t)}){|\psi_{m-1}(\mathbf{x},t)|}^{\gamma _{2m}},
		\end{align}
		where $\psi_{m-1}(\mathbf{x},t)$ is given by \eqref{sets}.
		\end{definition}
\begin{remark}
Given a TCHCBF $\mathfrak{h}(\mathbf{x},t)$ and a control signal $\mathbf{u}(\mathbf{x})$ that provides fixed-time convergence to the set $\mathfrak{C}_m$ and renders the system \eqref{cont_sys} forward complete \cite[Theorem  II.1]{kawan2021lyapunov}, it follows directly from Proposition \ref{prop1} that the set $\mathfrak{C}$ is convergent and forward-invariant.
\end{remark}
Next, we use the introduced TCHCBFs  to derive similar results to Section \ref{solution3} for $2^{nd}$ order leader-follower networks.
\subsection{Second order leader-follower multi-agent systems}
Consider a group of agents with $2^{nd}$ order dynamics as in \eqref{agent_fe}, under the task $\phi^d$. 
	We will formulate a quadratic program that renders   the set $\mathfrak{C}^d:=\cap_{k=1}^{2} \mathfrak{C}_{k}^d\subset\mathcal{S}^d\subset\mathbb{R}^{2n}$ corresponding to functions $\mathfrak{h}^d(x^d,t)$ and $\psi_1(x^d,t)$, defined by \eqref{safe_sets2}, \emph{robust convergent}, under the following Assumption.
\begin{assumption}\label{nost_2}
	Consider the $2^{nd}$ order leader-follower network \eqref{agent_fe} with the leader $i=n$. There exists a positive constant $\delta^d$ satisfying $	\|\sum\nolimits_{i\in\mathcal{N}_{n}, j\notin\mathcal{N}_{n}} \frac{{\partial \psi_1({{ x}^d},t)}}{{\partial {x_{j}^d}}} {\mathfrak{f}_{j}^d({x^d})}+\frac{{\partial \psi_1({{ x}^d},t)}}{{\partial {x_{i}^d}}} {\mathfrak{f}_{i,j}^d({x_{i}^d,x_{j}^d})}\|\leq \delta^d$, $\forall(x^d,t)\in\mathcal{S}^d\times[\tau_{l},\tau_{l+1})$, ${l\in\{0,\cdots,p^s-1\}})$.
	\end{assumption}
	In view of Assumption \ref{bound}, \eqref{sets} and the user defined function $\lambda_1(\cdot)$, Assumption \ref{nost_2} is feasible, too. In addition, there is no need for the leader to know $\delta^d$.

In the following, a control input $u_{n}$ will be found such that for all initial conditions $x^d(t_0)$, and under Assumption \ref{nost_2}, the trajectories of \eqref{agent_fe} converge to a the set $\mathfrak{C}_{1,rf}^d(t)\supset\mathfrak{C}^d(t)$ and $\psi_1( x^d,t)\in \mathfrak{C}_{2,rf}^d$, $\mathfrak{C}_{2,rf}^d(t)\supset \mathfrak{C}^d(t)$, in a fixed-time $t\leq T^d+t_0$, $T^d>0$. The sets $\mathfrak{C}_{1,rf}^d$ and $\mathfrak{C}_{2,rf}^d$ will be characterised in the sequel.
	
\textbf{QP formulation:} Define $z^d=\left[u_{n}, \varepsilon^d\right]^T\in\mathbb{R}^{2}$, and consider the following optimization problem.
\begin{align*}
&\mathop {\min }\limits_{u_{n}\in\mathbb{R},\varepsilon^d\in\mathbb{R}_{\geq 0}} \frac{1}{2}{{z^d}^T}z^d\label{B2}
\end{align*}\\
\rm{s.t.}\;\;\;
	\begin{equation}\label{CLF}
					\begin{array}{l}\sum\nolimits_{i\in\mathcal{N}_{n}} \frac{{\partial\psi_1({{ x}^d},t)}}{{\partial {x_{i}^d}}} {\mathfrak{f}_{i,i}^d({x_{i}^d})}+ \frac{{\partial\psi_1({{ x}^d},t)}}{{\partial {x_{i}^d}}} {\mathfrak{f}_{n,i}^d({x_{n}^d,x_{i}^d})}\\ + \frac{{\partial \psi_1({{ x}^d},t)}}{{\partial {x_{n}^d}}}{\mathfrak{g}_{n}^d}({x_{n}^d}){u_{n}}+ \frac{{\partial\psi_1({{ x}^d},t)}}{{\partial {x_{n}^d}}} {\mathfrak{f}_{n,n}^d({x_{n}^d})}\\+\frac{{\partial\psi_1({{ x}^d},t)}}{{\partial t}}\ge - {\alpha_2^d}\;\emph{\sgn}({	\psi_1({{ x}^d},t)}){	|\psi_1({x^d},t)|}^{\gamma _{12}^d}\\ \;\;\;\;\;\;\;\;\;\;\;\;\;\;\;\;\;\;\;\;-{\beta_2^d}\;\emph{\sgn}({	\psi_1({x^d},t)}){	|\psi_1({x^d},t)|}^{\gamma _{22}^d}-\varepsilon^d,
					\end{array}
					\end{equation}
					 where ${\alpha_2^d}>0, {\beta_2^d}>0$, $0<\gamma_{12}^d<1$,
$\gamma_{22}^d>1$.
		\begin{theorem}\label{Th2}
		Consider a given TCHCBF $\mathfrak{h}^d(x^d,t)$ from Definition \ref{TFHOCBF} with the associated functions $\psi_k(x^d,t)$, $k\in\{1, 2\}$, as defined in \eqref{sets}. Any control signal $u_{n}:\mathbb{R}\to\mathbb{R}$ which solves the quadratic program \eqref{CLF} renders the set  $\mathfrak{C}^d(t)$ robust convergent for the leader-follower network \eqref{agent_fe}, under Assumption \ref{nost_2}.
		\end{theorem}
		\begin{proof}
		In view of Theorem \ref{lem11}, constraint \eqref{CLF} corresponds to the fixed-time convergence of the closed-loop trajectories of network \eqref{agent_fe} to the set $\mathfrak{C}_{2,rf}^d(t):=\{ x^d\in\mathbb{R}^{2n}|\psi_1( x^d,t)\geq -\epsilon_{\max}^d\}$, where $\epsilon_{\max}^d$ is defined by the same formulation as in \eqref{D3}, within the fixed-time $T^d$ with similar expression as in \eqref{T}, built by parameters ${\alpha_2^d}, {\beta_2^d}>0$, $\gamma_{12}^d=1-\frac{1}{\mu^d}$,
$\gamma_{22}^d=1+\frac{1}{\mu^d}$, $\mu^d>1$, $k^d>1$ and $\delta^d$. These parameters are substitutions of ${\alpha^s}, {\beta ^s}$, $\gamma_{1}^s$,
$\gamma_{2}^s$, $\mu^s$, $k^s$ and $\delta^s$, respectively, in \eqref{T} and \eqref{D3}.
Then, according to \eqref{sets}, we get
$\dot{\mathfrak{h}}^d({ x}^d,t)+\lambda_1(\mathfrak{h}^d({ x}^d,t))\geq -\epsilon_{\max}^d$. 
Let $\lambda_1(\cdot)$ a linear extended class $\mathcal{K}$ function. Inspired by the notion of \emph{input-to-state safety} \cite{kolathaya2018input} and using the Comparison Lemma \cite[Lemma 3.4]{khalil2002nonlinear}, the set $\mathfrak{C}_{1,ref}^d(t):=\{ x^d\in\mathbb{R}^{2n}|\mathfrak{h}^d(x^d,t)\geq \lambda_1^{-1}(-\epsilon_{\max}^d)\}$, $t\geq T_e^d+t_0$ is forward-invariant and convergence of $\mathfrak{h}^d({ x}^d,t)$ to this set is achieved asymptotically.
Moreover, $\varepsilon^d>0$ relaxes \eqref{CLF} in the presence of conflicting specifications and its minimization results in a least violating solution to ensure the feasibility of \eqref{CLF}.
		\end{proof}
			\begin{figure*}
		\centering
		\hspace*{-0.5cm}
		\psfragscanon 
		\begin{subfigure}[b]{0.65\columnwidth}
		\psfrag{data1ps}[Bc][B1][0.45][0]{$p_3-p_1$}
		\psfrag{data2ps}[Bc][B1][0.45][0]{$p_3-p_2$}
		\includegraphics[width=6cm]{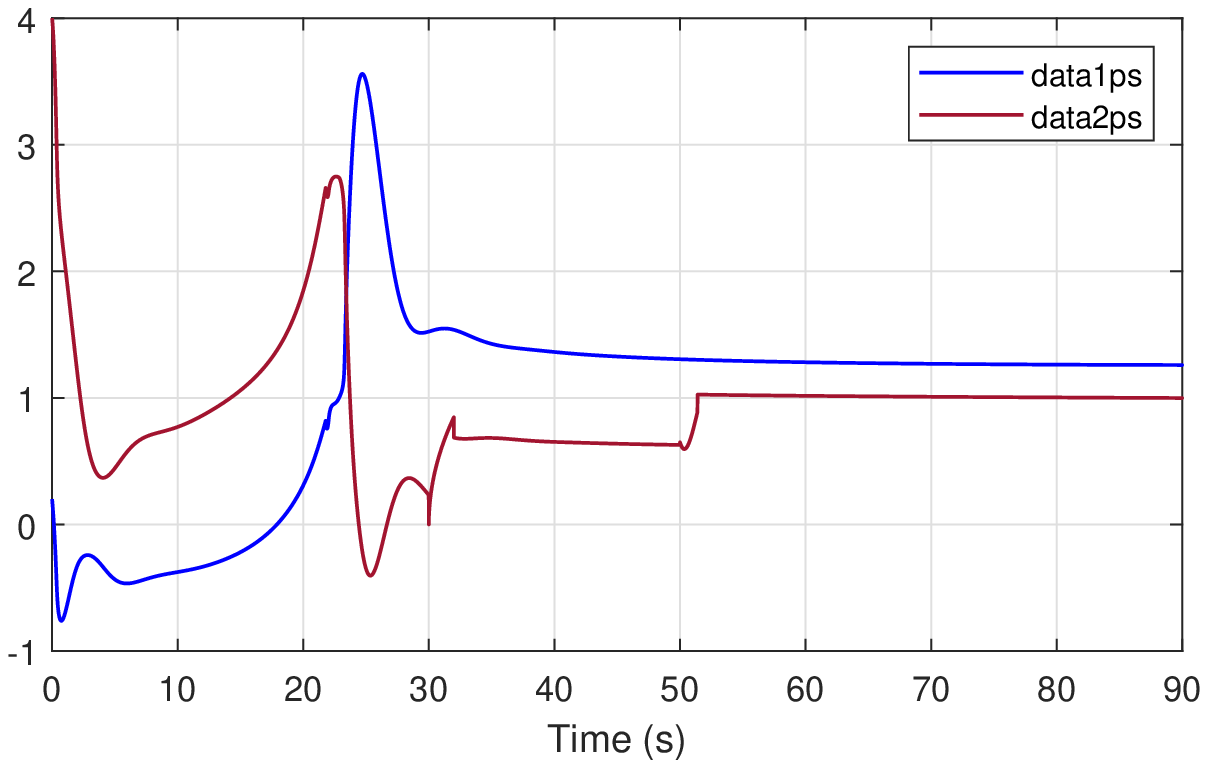}
		\caption{Position errors}
		\label{p_nostar}
		\end{subfigure}
		\begin{subfigure}[b]{0.65\columnwidth}
		\psfrag{data1vs}[Bc][B1][0.45][0]{$v_3-v_1$}
		\psfrag{data2vs}[Bc][B1][0.45][0]{$v_3-v_2$}
		\includegraphics[width=6cm]{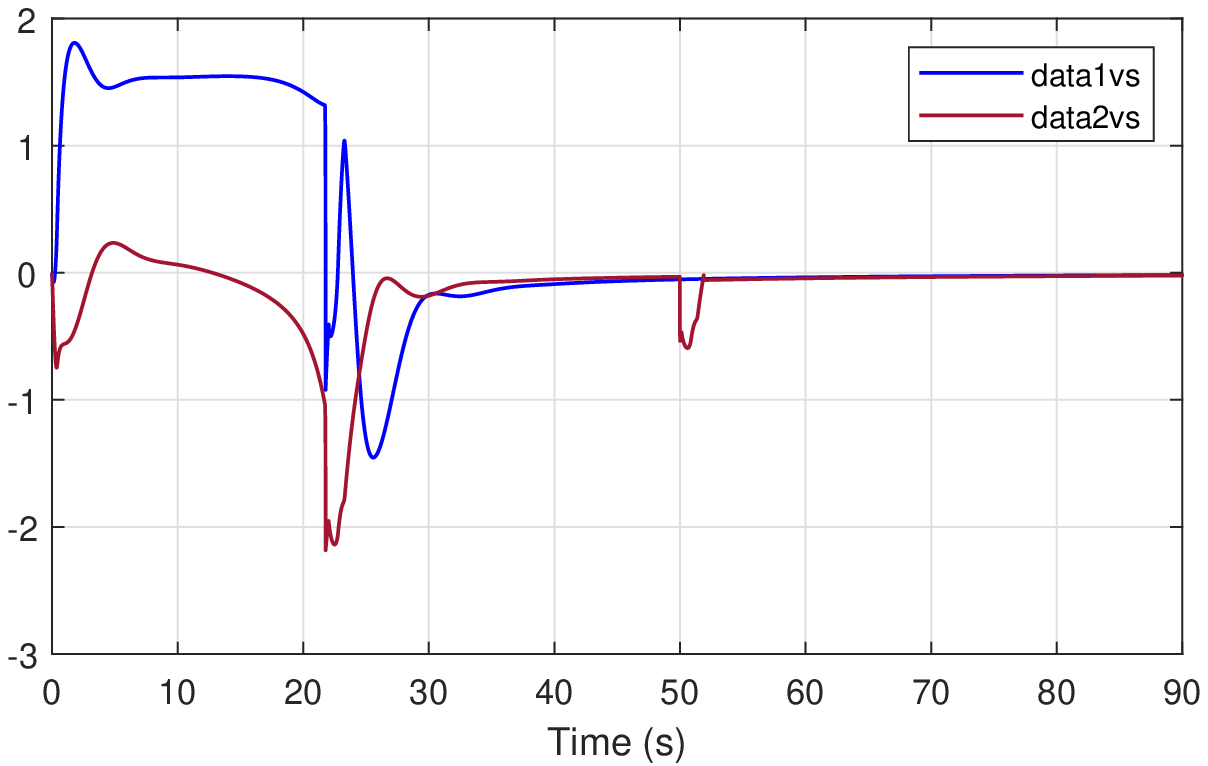}
		\caption{Velocity errors}
		\label{barrier}
		\end{subfigure}
		\begin{subfigure}[b]{0.65\columnwidth}
		\psfrag{data11111}[Bc][B1][0.43][0]{$\mathfrak{h}_1^d(x_1^d,t)$}
		\psfrag{data22222}[Bc][B1][0.43][0]{$\mathfrak{h}_2^d( x_2^d,t)$}
		\psfrag{data33333}[Bc][B1][0.43][0]{$\mathfrak{h}_3^d( x_3^d,t)$}
		\includegraphics[width=6cm]{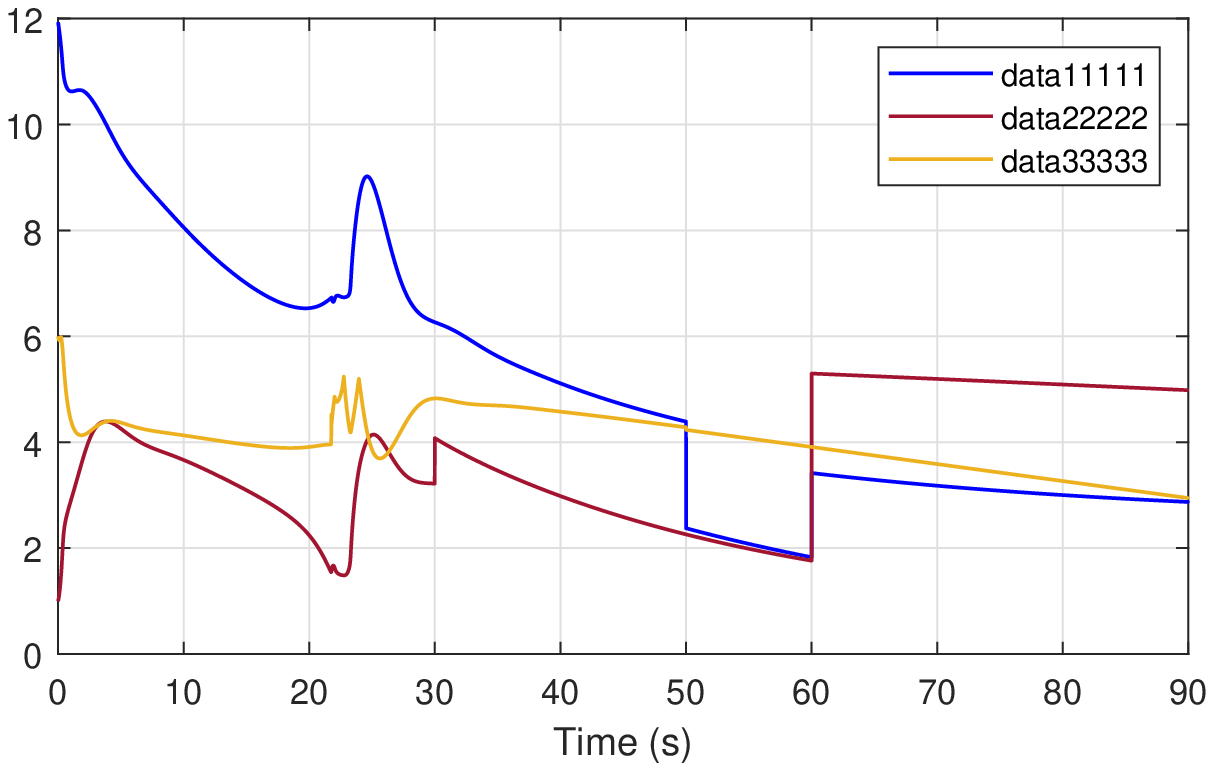}
		\caption{TCHCBFs}
		\label{barrier}
		\end{subfigure}
		 \psfragscanoff
		 		\caption{Leader-follower network \eqref{agent_fe} under full information of the leader from the network.}\label{fstar}
	\end{figure*}
			\begin{figure*}
		\psfragscanon 
		\begin{subfigure}[b]{0.65\columnwidth}
		\psfrag{datap1}[Bc][B1][0.45][0]{$p_3-p_1$}
		\psfrag{datap2}[Bc][B1][0.45][0]{$p_3-p_2$}
		\includegraphics[width=6cm]{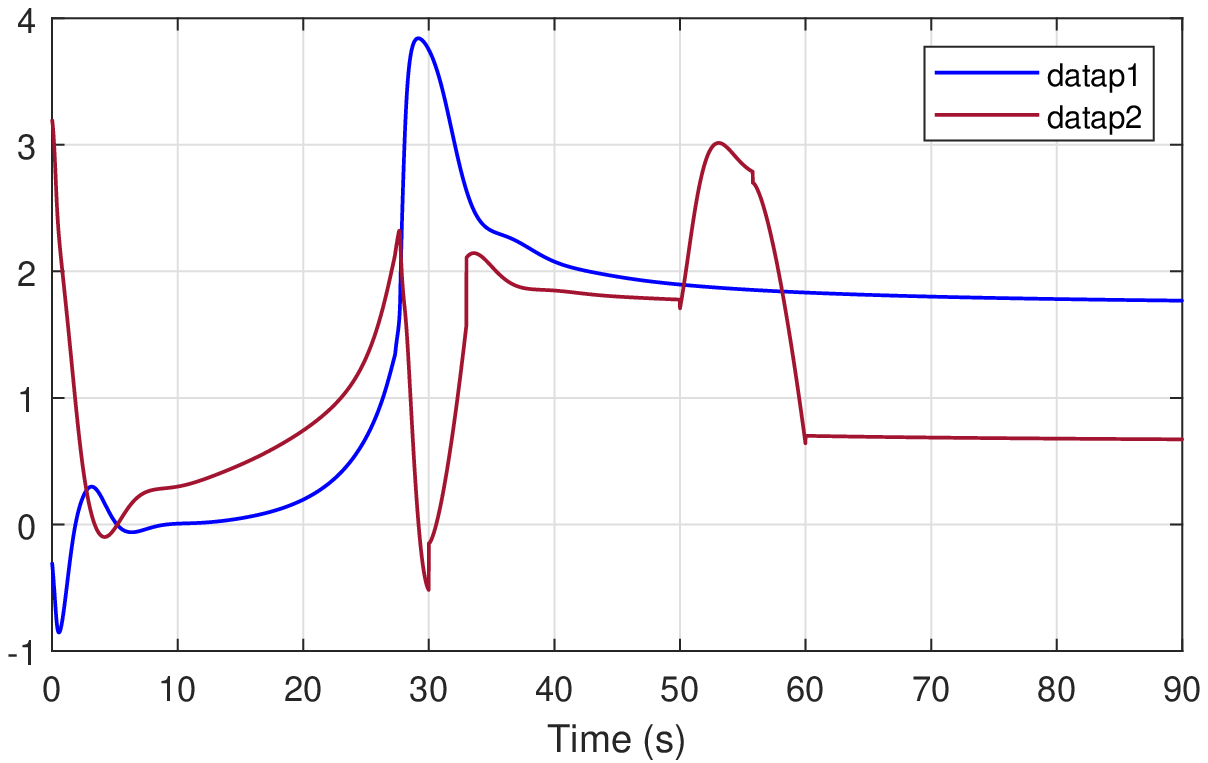}
		\caption{Position errors}
		\label{p_nostar}
		\end{subfigure}
		\begin{subfigure}[b]{0.65\columnwidth}
		\psfrag{data1}[Bc][B1][0.45][0]{$v_3-v_1$}
		\psfrag{data2}[Bc][B1][0.45][0]{$v_3-v_2$}
		\includegraphics[width=6cm]{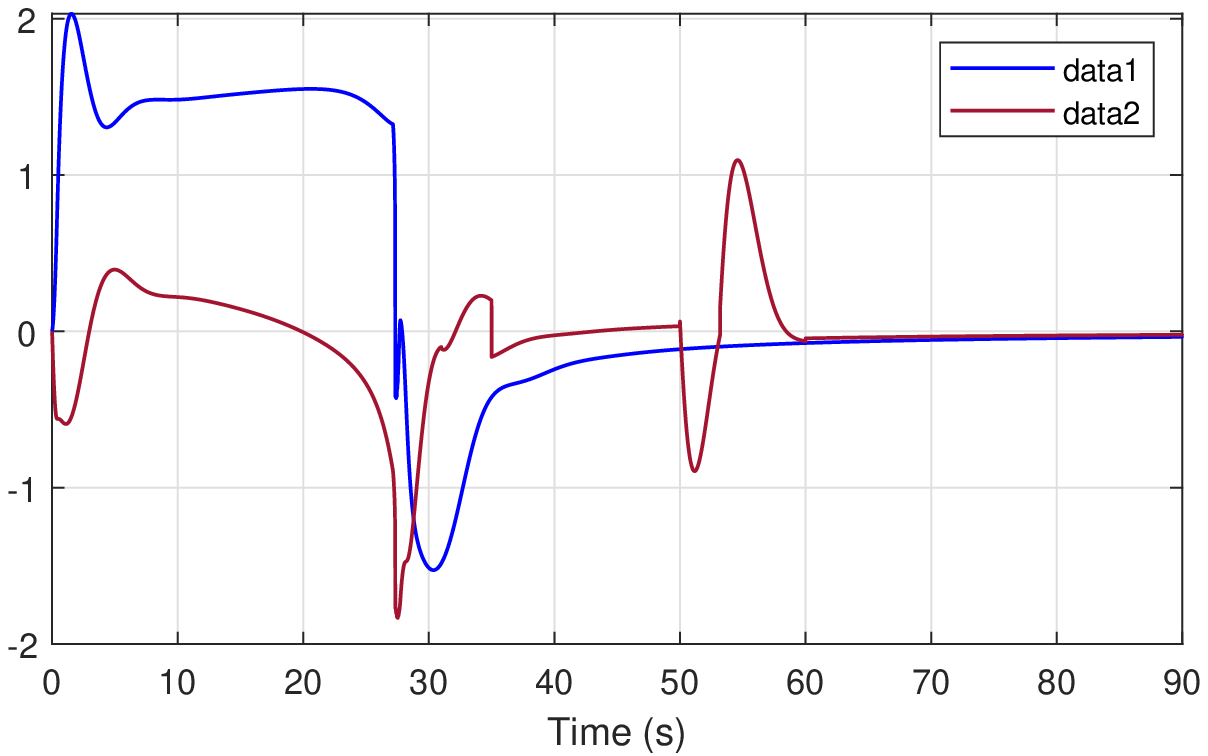}
		\caption{Velocity errors}
		\label{barrier}
		\end{subfigure}
		\begin{subfigure}[b]{0.65\columnwidth}
		\psfrag{dataaaaaa}[Bc][B1][0.43][0]{$\mathfrak{h}_1^d(x_1^d,t)$}
		\psfrag{databbbbb}[Bc][B1][0.43][0]{$\mathfrak{h}_2^d( x_2^d,t)$}
		\psfrag{dataccccc}[Bc][B1][0.43][0]{$\mathfrak{h}_3^d( x_3^d,t)$}
		\includegraphics[width=6cm]{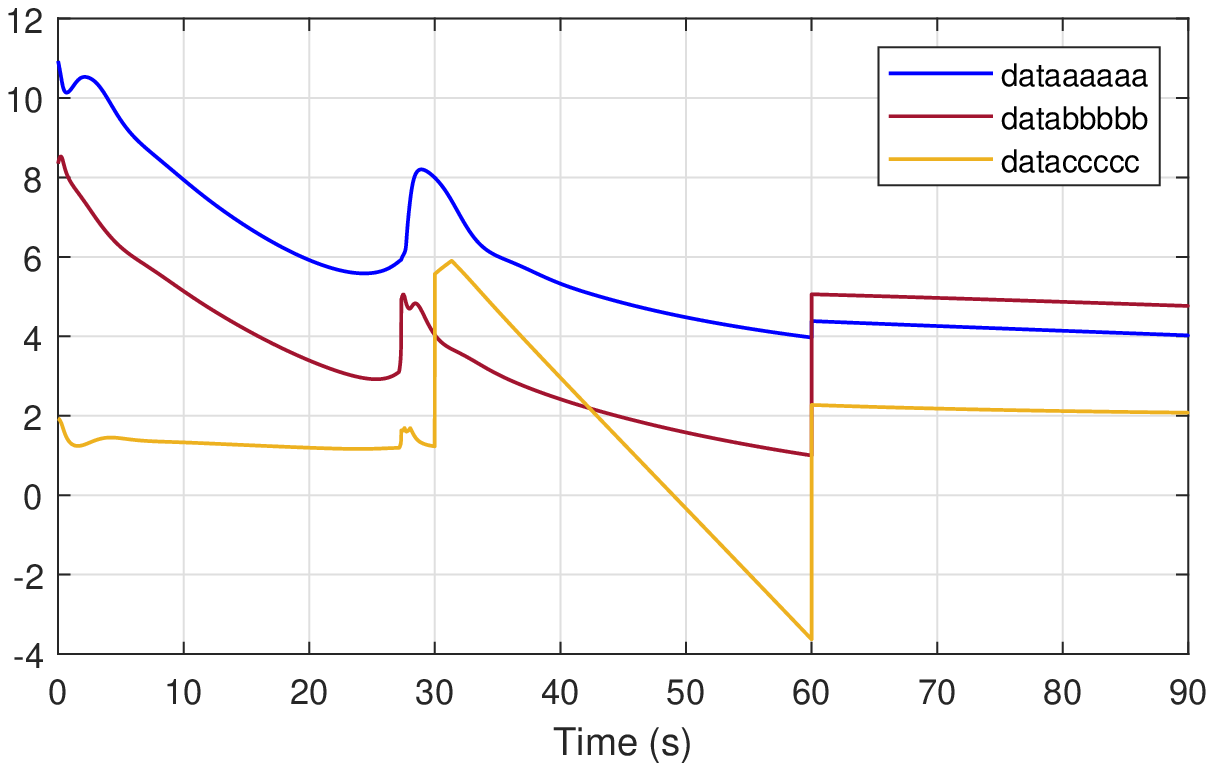}
		\caption{TCHCBFs}
		\label{barrier}
		\end{subfigure}
		 \psfragscanoff
		 		\caption{Leader-follower network \eqref{agent_fe} under Assumption \ref{nost_2}}\label{fnostarrr}
	\end{figure*}
	\begin{corollary}
	Consider TCHCBF $\mathfrak{h}^d(x^d,t)$ from Definition \ref{TFHOCBF} with the associated functions $\psi_k(x^d,t)$, $k\in\{1, 2\}$, as defined in \eqref{sets} for network \eqref{agent_fe}. Then, any control signal $u_{n}$ satisfying
	\begin{equation*}
					\begin{array}{l}	\sum\nolimits_{i\in\mathcal{V}} {\frac{{\partial\psi_1({{ x}^d},t)}}{{\partial {x_{i}^d}}} {\mathfrak{f}_i^d({x^d})}  + \frac{{\partial \psi_1({{ x}^d},t)}}{{\partial {x_{n}^d}}}{\mathfrak{g}_{n}^d}({x_{n}^d}){u_{n}}}\\+ \frac{{\partial\psi_1({{ x}^d},t)}}{{\partial t}}\ge - {\alpha_2^d}\;\sgn({\psi_1({{ x}^d},t)}){	|\psi_1({x^d},t)|}^{\gamma _{12}^d}\\ \;\;\;\;\;\;\;\;\;\;\;\;\;\;\;\;\;\;\;\;-{\beta_2^d}\;\sgn({	\psi_1({x^d},t)}){	|\psi_1({x^d},t)|}^{\gamma _{22}^d},
					\end{array}
					\end{equation*}
	for constants $\alpha_2^d>0$, $\beta_2^d>0$, $0<\gamma_{12}^d<1$, $\gamma_{22}^d>1$, renders the set $\mathfrak{C}^d:=\cap_{k=1}^{2} \mathfrak{C}_{k}^d\subset\mathcal{S}^d\subset\mathbb{R}^{2n}$ convergent and forward invariant. Moreover, it holds that $x^d\models\phi^d$ within $T^d\leq \frac{1}{\alpha_2^d(1 - \gamma_{12}^d)}+\frac{1}{\beta_2^d(\gamma_{22}^d-1)}$.
	\end{corollary}
	\begin{proof}
	Follows by the proof of Lemma \ref{lem1} with incorporating the arguments in Proposition \ref{prop1}.
	\end{proof}
\begin{remark}\label{singularity}
Note that there might exist singularities in the solution of \eqref{CLF} in some points; in particular, whenever $\frac{{\partial {\psi_1({{x^d}},t)}}}{{\partial {x_{n}^d}}}{g}_{n}^d({x_{n}^d})=0$. 
Under the assumption that singularities lie inside the safe sets, it can be shown that the required inequalities remain feasible and can be satisfied \cite[Proposition 4]{tan2021high}.
\end{remark}
\section{Simulations}\label{sim}
Consider a leader-follower multi-agent system consisting of $M:=3$ second order dynamics agents.  We consider dependent tasks, where the third agent acts as the leader. 
Consider the formula $\phi^d=\phi_1^d\wedge\phi_2^d\wedge\phi_3^d$ with $\phi_1^d:=  G_{\left[ {10,30} \right]}(|{{v_3} - {v_2}}| \le {2})\wedge F_{\left[ {10,90} \right]}(|{{p_1} + 1 - {p_3}}| \le {1})$, $\phi_2^d:= F_{\left[ {10,30} \right]}(|{{v_3}-v_2 }| \le {1})\wedge G_{\left[ {30,90} \right]}(| {{v_1} - {v_3}}| \le {2})$, $\phi_3^d:=F_{\left[ {10,60} \right]}(|{{v_3} - {v_1}}| \le {1})\wedge G_{\left[ {60,90} \right]}(| {{v_2} - {v_3}}| \le {1})\wedge G_{\left[ {50,60} \right]}(|{{p_2}+1 - {p_3}}| \le {1})$.
As the position dependent formulas are of relative degree $2$, we use TCHCBFs of order $m=2$. Furthermore, TCHCBFs of order $m=1$ are considered for velocity dependent specifications. 
We choose the parameters of the QP formulation as $\mu^d=2$, $\alpha_2^d=\beta_2^d=1$, and $\lambda_1(r):=r$.
We focus on the effect of leader agent information on the group task satisfaction.
First, we consider the network \eqref{agent_fe}, where $L:=\left[ {\begin{array}{*{20}{c}}
{1}&{0}&{-1}\\
{0}&{1}&{-1}\\
{0}&{0}&{0}
\end{array}} \right]$ and the input matrix $\mathfrak{g}^d(x^d):=\left[ {\begin{array}{*{20}{c}}
0_{1\times 5},1
\end{array}} \right]^T$, where leader has knowledge of the functions $ \frac{{\partial\psi_1({x^d},t)}}{{\partial {x_i^d}}}$ and dynamics ${\mathfrak{f}_{i}^d({x^d})}$, ${i\in\left\{1,\cdots,n\right\}}$ (an equivalent condition to Assumption \ref{star} for $2^{nd}$ order dynamics), where the convergent and forward invariance property of set $\mathfrak{C}^d(t)$, as well as task satisfaction are concluded as shown in Fig. \ref{fstar}.
Next, we consider the agent $i=1$ as the leader's neighbor and $i=2$ as a neighbor to $i=1$ under Assumption \ref{nost_2}, where $\mathfrak{f}^d(x^d):=\left[ {\begin{array}{*{20}{c}}
{{0_3}}&{{I_3}}\\
{ - L}&{ - L}
\end{array}} \right]x^d$ with $L:=\left[ {\begin{array}{*{20}{c}}
{2}&{-1}&{-1}\\
{-1}&{1}&{0}\\
{0}&{0}&{0}
\end{array}} \right]$. By solving \eqref{CLF} where $\delta^d=2.86$, the fixed-time convergence to the set $\mathfrak{C}_{2,ref}^d(t)\supset \mathfrak{C}^d(t)$ is achieved with $\epsilon_{max}^d=6.01$ using \eqref{D3}, which gives $\mathfrak{h}^d(x^d,t)\geq \lambda_1^{-1}(-\epsilon_{\max}^d)=-6.01$. Fig. \ref{fnostarrr} shows a violation in satisfaction of the third task in $t\in \left[50,60\right]$ which certifies this result, although it is less conservative than the estimation.
The computation times on an Intel Core i5-8365U with 16 GB of RAM are about $2.1$ms.
\section{Conclusion}\label{conc}
Based on a class of time-varying convergent higher order control barrier functions, we have presented feedback control strategies to find solutions for the leader-follower multi-agent systems performance under STL tasks, based on the leader's knowledge on the followers' states. The finite convergence time is characterized independently of the initial conditions of the agents. 
Future work will extend these results to decentralized barrier certificates in networks containing more than one leader.

\bibliographystyle{IEEEtran}
\bibliography{ref}

\end{document}